\documentclass[5p,times,sort&compress]{elsarticle}
\pdfoutput=1
\usepackage[utf8]{inputenc}
\usepackage[T1]{fontenc}
\usepackage{fullpage}
\usepackage{microtype}
\usepackage{amsmath,amssymb,amsthm}
\usepackage{hyphenat}
\usepackage{natbib} %

\usepackage{pgfplots}
\usepackage{hyphenat}
\usepackage{todonotes}
\usepackage{doi}
\usepackage{paralist}
\usepackage{subcaption}
\usepackage{booktabs}
\usepackage{multirow}
\usepackage[titlenumbered,vlined,linesnumbered,ruled]{algorithm2e}

\usepackage[sort&compress,nameinlink,noabbrev,capitalize]{cleveref} 

\pgfplotsset{compat=1.3}
\hypersetup{colorlinks,allcolors=blue}

\newcommand{\poly}{\ensuremath{\text{poly}}}
\newcommand{\NP}{\ensuremath{\text{NP}}}
\newcommand{\coNP}{\ensuremath{\text{coNP}}}
\newcommand{\huawei}{\textit{Huawei}}
\newcommand{\azure}{\textit{Azure}}

\newcommand{\unlessPK}{\ensuremath{\coNP\subseteq \NP/\poly}}

{}
\newcommand{\ntypes}{\tau}
\theoremstyle{definition}
\newtheorem{theorem}{Theorem}[section]
\newtheorem{example}[theorem]{Example}

\newtheorem{lemma}[theorem]{Lemma}
\newtheorem{problem}[theorem]{Problem}
\newtheorem{definition}[theorem]{Definition}
\newtheorem{proposition}[theorem]{Proposition}
\newtheorem{conjecture}[theorem]{Conjecture}

\newtheorem{rrule}[theorem]{Reduction Rule}
\newcommand{\ai}[1]{\ensuremath{a^{(#1)}}}
\newcommand{\ri}[1]{\ensuremath{s^{(#1)}}}
\newcommand{\di}[1]{\ensuremath{e^{(#1)}}}
\newcommand{\N}{\ensuremath{\mathbb N}}
\newcommand{\mc}{\ensuremath{h}}
\newcommand{\numflavors}{\ensuremath{\phi}}
\newcommand{\OPT}{\text{OPT}}
\newcommand{\alphagap}{$\alpha$\hyp gap}

\newtheoremstyle{iostuff}%
{0pt}%
{0pt}%
{\hangindent=\parindent}%
{}%
{\itshape}%
{:}%
{.5em}%
{}%

\theoremstyle{iostuff}
\newtheorem*{probinstance}{Input}
\newtheorem*{probfind}{Find}

\crefname{construction}{Construction}{Constructions}
\crefname{oquestion}{Open Question}{Open Question}
\crefname{paragraph}{Section}{Sections}
\crefname{lemma}{Lemma}{Lemmas}
\Crefname{lemma}{Lem.}{Lem.}
\crefname{theorem}{Theorem}{Theorems}
\Crefname{theorem}{Thm.}{Thm.}
\crefname{proposition}{Proposition}{Propositions}
\Crefname{proposition}{Prop.}{Props.}
\crefname{remark}{Remark}{Remarks}
\Crefname{remark}{Rem.}{Rem.}
\crefname{prop}{Property}{Properties}
\crefname{problem}{Problem}{Problems}
\crefname{observation}{Observation}{Observations}
\crefname{corollary}{Corollary}{Corollaries}
\crefname{example}{Example}{Examples}
\Crefname{corollary}{Cor.}{Cors.}
\crefname{line}{line}{lines}
\crefname{section}{Section}{Sections}
\Crefname{section}{Sec.}{Secs.}
\crefname{rrule}{Reduction Rule}{Reduction Rules}

\begin{document}
\begin{frontmatter}
  \author{Ren\'e van Bevern\corref{cor1}}
  \ead{rene.van.bevern@huawei.com}
  \author{Andrey Melnikov}
  \ead{melnikov.andrey@huawei.com}
  \author{Pavel\ V.\ Smirnov}
  \ead{smirnov.pavel2@huawei.com}
  \author{Oxana\ Yu.\ Tsidulko}
  \ead{tsidulko.oksana@huawei.com}
  \address{Huawei Technologies Co., Ltd.}
  \title{On data reduction
    for dynamic vector
    bin packing}

  \cortext[cor1]{Corresponding author}
  
  \begin{abstract}
    \looseness=-1
    We study a
    dynamic vector bin packing (DVBP) problem.
    We show hardness
    for shrinking arbitrary DVBP instances
    to size polynomial in the number of request types or
    in the maximal number of requests overlapping in time.
    We also present a simple polynomial\hyp time data reduction algorithm
    that allows to recover $(1+\varepsilon)$\hyp approximate
    solutions for arbitrary~$\varepsilon>0$.
    It shrinks instances
    from  Microsoft Azure and Huawei Cloud by an order
    of magnitude
    for $\varepsilon=0.02$.
  \end{abstract}

\end{frontmatter}

\thispagestyle{empty}

\section{Introduction} \label{dr:sec:reduction}
\noindent
Motivated by applications to computer storage allocation,
\citet{CGJ83} introduced the %
dynamic bin packing problem (DBP).
Motivated by resource allocation problems
in cloud data centers,
we study the multi\hyp dimensional generalization
of the problem:

\begin{problem}[Dynamic vector bin packing (DVBP)]
  \label[problem]{dr:prob:nvmsp}
  \begin{probinstance}
    A bin capacity~$b\in\N^d$,
    \emph{requests}~$\ai {1},\dots,\ai {n}\in \N^d$, 
    request start times~$\ri 1,\dots,\ri n\in\N$,
    request end times~$\di 1,\dots,\di n\in\N$. 
  \end{probinstance}
  \begin{probfind}
    A partition of~$\{1,\dots,n\}$
    into \emph{bins} $B_1,\dots,B_k$ %
    with minimum~$k$ such that,
    for any \emph{time instant}~$t\in \N$ and each bin~$B_j$,
    it holds that (component\hyp wise)
    \[\sum_{\substack{i\in B_j\\
          \ri i\leq t<\di i }} \ai i
  \le b.\]
  \end{probfind}
\end{problem}

\begin{definition}[flavor, type, height]
  Two requests $a^{(i)}$ and $a^{(j)}$
  have the same \emph{flavor} if
  $a^{(i)}=a^{(j)}$.
  They are of the same \emph{type}
  if, additionally,
  $\ri{i}=\ri{j}$ and $\di{i}=\di{j}$.

  The \emph{height}~$\mc$ of an instance
  is the maximum number of requests active
  at the same time.
  We denote the number of flavors by~$\numflavors$,
  and the number of types by~$\ntypes$.
\end{definition}
\begin{example}\label{ex:vm-scheduling}
  Consider a pool of
  identical servers of capacity~$b\in \N^2$,
  providing $b_1$ processors and $b_2$ units of memory.
  Customers make requests $\ai i\in\N^2$
  for virtual machines
  with $\ai i_1$ processors and $\ai i_2$ units of memory
  from time \ri i to \di i.
  The problem of satisfying
  all customer requests
  using a minimum number of servers
  is DVBP with $d=2$.
  In practice,
  customers usually have the choice
  among a few virtual machine flavors.
\end{example}

\noindent
In the virtual machine assignment scenario,
the start and end times of requests
are usually not known beforehand
and the problem
has to be solved online:
on arrival, each request
is immediately assigned to a bin,
reassigning requests to other bins may be allowed or not \citep{CKPT17}.
Since, in the worst case,
any online algorithm for DVBP
will use $\Omega(d^{1-\varepsilon})$~times
the optimal number of bins,
cloud providers use various online heuristics tailored
to their presumed input distributions \citep[see, e.\,g.,][]{SLD+18}.

In order to empirically evaluate the quality of these heuristics,
computing (close to) optimal solutions to the offline DVBP
is desirable.
Unfortunately,
even when all requests have identical start and end times
and $d=2$,
the best known polynomial\hyp time approximation algorithm
for DVBP yields an asymptotic approximation factor
of $1.405+\varepsilon$ \citep{BEK16}.
Unless P${}={}$NP,
asymptotic $(1+\varepsilon)$\hyp approximations
for arbitrarily small $\varepsilon>0$
will require superpolynomial time \citep{Woe97b,Ray21}.

Data reduction has proven to be a powerful way
to cope with superpolynomial problem complexity
\citep{ABG+20,ALM+22}.

\paragraph{Our results and outline of this work}
%
We study the potential for polynomial\hyp time
data reduction for DVBP.

In \cref{dr:prel},
we introduce basic tools and notation. %
In \cref{dr:sec:hard},
we show hardness results for data reduction. %
In \cref{sec:ilp},
we present exact data reduction,
and in \cref{sec:dr} --- data reduction
that guarantees recoverability of $(1+\varepsilon)$\hyp approximate
solutions for any~$\varepsilon>0$.

Due to our hardness results,
we cannot prove bounds on the data reduction effect.
For this reason,
the effect is evaluated experimentally in \cref{dr:exp}.
We were able to shrink the number of request types
in real\hyp world instances to about $2\,\%$ of their initial number,
and the number of requests --- to~$8\%$ on one data set,
and to~$20\%$ on another.

\section{Preliminaries}
\label{dr:prel}

\subsection{Parameterized optimization and decision problems}

\begin{definition}
	A \emph{decision problem} is a subset~$\Pi\subseteq\Sigma^*$
	for some finite alphabet~$\Sigma$.
	The task is,
	given an \emph{instance}~$x\in\Sigma^*$,
	determining whether $x\in \Pi$.
	If $x\in\Pi$,
	then $x$~is a \emph{yes\hyp instance}.
	Otherwise,
	it is a \emph{no\hyp instance}.
\end{definition}

\noindent
For optimization problems,
we use the terminology of Garey and Johnson \cite{GJ79}.
We will only consider \emph{minimization problems} in our work.
\begin{definition}
  A \emph{combinatorial optimization problem}~$\Pi$ is a triple~$\Pi=(D_\Pi,S_\Pi,m_\Pi)$,
  where
  \begin{enumerate}
  \item $D_\Pi$ is a set of \emph{instances},
  \item $S_\Pi$ is a function
    assigning to each instance $I\in D_\Pi$
    a finite set~$S_\Pi(I)$ of
    \emph{(feasible) solutions}, and
  \item $m_\Pi$ is a function
    assigning a \emph{solution cost}~$m_\Pi(I,\sigma)$
    to each feasible solution~$\sigma\in S_\Pi(I)$
    of an instance~$I\in D_\Pi$.
  \end{enumerate}
  An \emph{optimal solution}
  for an instance~$I\in D_\Pi$
  is a feasible solution~$\sigma\in S_\Pi(I)$
  minimizing~$m_\Pi(I,\sigma)$.
  Its cost is denoted as $\OPT_\Pi(I)$,
  where we drop the subscript~$\Pi$
  when the optimization problem is clear from context.
  An \emph{$\alpha$\hyp approximate solution} for an instance~$I$
  of a combinatorial optimization problem~$\Pi$
  is a feasible solution of cost at most $\alpha\cdot\OPT_\Pi(I)$.
\end{definition}

\noindent
Each optimization problem comes with
natural associated
\emph{decision versions} and
\emph{gap\hyp version}:
\begin{definition}
  Let $\Pi$~be a combinatorial optimization problem
  and $\alpha\geq 1$.

  By \emph{$\alpha$\hyp gap $\Pi$}
  we denote the following decision problem:
  given a number~$r$ and an
  instance $I$ of~$\Pi$ 
  such that $\OPT(I)\leq r$ or $\OPT(I)>\alpha r$,
  it is required to decide whether $\OPT(I)\leq r$.

For $\alpha=1$,
we call $\alpha$\hyp gap $\Pi$ simply
the \emph{decision version of $\Pi$}.
\end{definition}

\noindent
If $\alpha$\hyp gap $\Pi$ is NP\hyp hard,
then $\Pi$ cannot be better than $\alpha$\hyp approximated in polynomial\hyp time,
unless P${}={}$NP.

\bigskip\noindent
We use the parameterized complexity notation
due to \citet{FG06},
as it equally well applicable
to decision and optimization problems.

\begin{definition}
  A \emph{parameterization} is a
  polynomial\hyp time computable
  mapping~$\kappa\colon\Sigma^*\to\N$
  of instances (of decision or optimization problems)
  to a \emph{parameter}.
  For a (decision or optimization) problem~$\Pi$
  and parameterization~$\kappa$,
  the pair $(\Pi,\kappa)$ is called
  a \emph{parameterized (decision or optimization) problem}.
\end{definition}

\subsection{Data reduction with performance guarantees}
\noindent
A \emph{data reduction}
is a conversion of a problem input
into a ``similar'' but smaller one.
There are multiple ways to formalize the term.
A well\hyp known notion of data reduction with performance guarantees
is kernelization~\citep{FG06}:

\begin{definition}
	\label{dr:def:compression}
	A~\emph{kernelization} for a parameterized decision problem~$(\Pi,\kappa)$
	is a polynomial\hyp time algorithm
	that maps any instance~$x\in\Sigma^*$
	to an instance~$x'\in\Sigma^*$
	such that
	\begin{enumerate}[(i)]
		\item $x\in \Pi \iff x'\in \Pi$, and
		\item $|x'|\leq g(\kappa(x))$ for some computable function~\(g\).
	\end{enumerate}
	We call \(x'\) the \emph{problem kernel}
	and \(g\) its \emph{size}.
\end{definition}

\noindent
We will also
consider approximate data reduction \cite{LPRS17}:

\begin{definition}
  An \emph{$\alpha$-approximate preprocessing}
  for a parameterized optimization problem~$(\Pi,\kappa)$
  consists of two algorithms:
  \begin{enumerate}[(i)]
  \item The first algorithm reduces an instance~$I$
    of~$\Pi$ to an instance~$I'$ of~$\Pi$
    in polynomial time,
  \item The second algorithm turns any
    \(\beta\)\hyp approximate solution for~$I'$
    into an \(\alpha\beta\)\hyp
    approximate
    solution for~$I$ in polynomial time.
  \end{enumerate}
  An \emph{$\alpha$-approximate kernelization}
  is an $\alpha$-approximate preprocessing
  such that $|I'|\leq g(\kappa(I))$
  for some computable function~$g\colon\N\to\N$.
  We call $g$ the \emph{size} of the \emph{approximate kernel}~$I'$.
\end{definition}

\subsection{Hardness of data reduction}
\noindent
To exclude problem kernels of polynomial size,
we use \emph{AND\hyp compositions} \citep{BJK14}.

\begin{definition}[AND-composition]
	\label[definition]{dr:def:crossco}
	An equivalence relation over~$\Sigma^*$ is
        \emph{polynomial} if
	\begin{itemize}
        \item there is an algorithm that decides~$x\sim y$ in polynomial time
          for any two instances~$x,y\in\Sigma^*$, and %
        \item the number of equivalence classes of~$\sim$ over any \emph{finite} set~$S\subseteq
          \Sigma^*$ is polynomial in $\max_{x\in S}|x|$.
	\end{itemize}
	
	\noindent
	A language~$K\subseteq\Sigma^*$ \emph{AND\hyp composes} into a
	parameterized language~$(\Pi,\kappa)$ if there
	is a polynomial\hyp time algorithm,
	called \emph{AND\hyp composition},
	that, given $s$~instances $x_1,\ldots,x_s$ that are
	equivalent under some polynomial equivalence relation, outputs
	an instance~$x^*$ such that
	\begin{itemize}
        \item $\kappa(x^*)\in\poly(\max^p_{i=1}|x_i|+\log s)$,
        \item $x^*\in \Pi$
          if and only if $x_i\in K$ for all~$i\in\{1,\dots,s\}$.
	\end{itemize}
\end{definition}

\begin{proposition}[\citet{Dru12}]
	\label{dr:prop:crosscomp} 
	If an NP-hard language $K\subseteq \Sigma^*$ 
        AND\hyp composes into a
	parameterized problem~$(\Pi,\kappa)$,
	then
	there is no 	polynomial-size problem kernel for~$(\Pi,\kappa)$
        unless $\unlessPK$.
\end{proposition}

\section{Classification results}\label{dr:sec:hard}
\noindent
In this section,
we prove that,
unless the polynomial\hyp time hierarchy collapses,
DVBP has no problem kernels
of size polynomial in~$\mc+\numflavors$.
Our reduction gives good reason to conjecture
that there are no $\alpha$\hyp approximate kernels
for any $\alpha<600/599$, either.

\subsection{Existence of exponential\hyp size kernels}
\label{dr:sec:fpt}
\noindent
Before proving the non\hyp existence
of problem kernels with size polynomial in~$\mc+\numflavors$,
we prove that, in principle,
problem kernels (of exponential size) \emph{do} exist.

A folklore result from parameterized complexity \cite[Theorem~1.4]{FLSZ19}
is that the following
theorem gives a problem kernel of size $O(k^{2\mc})\subseteq O(\mc^{2\mc})$ for DVBP.
\begin{theorem}
  \label{dr:thm:mc}
  The DVBP problem can be solved in $O(k^{2\mc}\cdot n+n\log n)$  time.
\end{theorem}

\begin{proof}
  In $O(n\log n)$~time,
  we transform the instance
  so that $\ri i,\di i\in\{1,\dots,2n\}$
  for all $i\in\{1,\dots,n\}$
  (see \cref{dr:thm:tcri} in \cref{dr:timecomp} for details).
  We then solve the problem using dynamic programming.

  For each $t\in\{1,\dots,2n\}$,
  let $S_t:=\{ j\mid \ri j\leq t<\di j\}$
  be the set of requests active at time~$t$.
  For $t\in\{0,\dots,2n\}$
  and any partition~$X_1\uplus X_2 \uplus \dots \uplus X_{k} = S_t$
  (here, $\uplus$ denotes disjoint unions),
  consider
  the predicate
  \(
  T[t,X_1,X_2,\dots,X_{k}]
  \)
  that is true if and only if
  all requests~$j$ with $\ri j\geq t$ %
  can be packed into $k$~bins %
  so that requests $X_i$ are packed into bin~$i$
  for $i\in\{1,\dots,k\}$.
  Then $T[2n+1,X_1,X_2,\dots,X_{k}]$ is true
  for any choice
  of the sets~$X_1\uplus \dots\uplus X_{k}=S_{2n+1}=\emptyset$.

  Moreover,
  $T[t,X_1,X_2,\dots,X_{k}]$
  is true if and only if
  there is a partition
  $X_1'\uplus X_2'\uplus \dots\uplus X_{k}'=S_{t+1}$,
  such that $T[t+1,X_1',X_2',\dots,X_{k}']$ is true
  and %
  $X_i\cap S_t\cap S_{t+1}=X_i'\cap S_{t}\cap S_{t+1}$
  for each $i\in\{1,\dots,k\}$.
  That is, the partitions $\{X_i\}_{i=1}^k$
  and $\{X_i'\}_{i=1}^k$ put the requests
  in $S_t\cap S_{t+1}$ into the same bins.

  Thus, each of the $O(k^{\mc}n)$ values of~$T$
  can be computed in $O(k^{\mc})$~time
  by iterating over all
  possible partitions~$\{X_i'\}_{i=1}^k$,
  yielding a total running time of $O(k^{2\mc}n)$.
  The answer can be found in
  $T[0,\emptyset,\emptyset,\dots,\emptyset]$.
\end{proof}

\subsection{Non\hyp existence of polynomial\hyp size kernels}
\noindent
We now prove that,
unless the polynomial\hyp time hierarchy collapses,
DVBP has no problem kernels with size polynomial even in $\mc+\numflavors$.
Indeed,
not even $600/599$-gap DVBP has such a kernel: 


\begin{theorem}
    \label{dr:thm:no-alpha-gap-polykern}
    $\alpha$\hyp gap DVBP does not have a problem kernel
    of size polynomial in $\mc+\numflavors$,
    unless the polynomial hierarchy collapses,
    for any $\alpha\leq 600/599$.
\end{theorem}
\noindent
Note that the statement of
\cref{dr:thm:no-alpha-gap-polykern}
holds for \emph{any} parameter
bounded from above by $\mc + \numflavors$.
Moreover,
the theorem leads us to the following conjecture.
\begin{conjecture}
  DVBP has no $\alpha$\hyp approximate problem kernels
  of size polynomial in $\mc+\numflavors$,
  for any $\alpha<600/599$.
\end{conjecture}
\noindent
Unfortunately,
we cannot formally connect this conjecture,
for example,
to a collapse of the 
polynomial\hyp time hierarchy,
the main obstacle being the
final open question
in the work of \citet{LPRS17}.

In order to prove \cref{dr:thm:no-alpha-gap-polykern},
we prove that the following problem AND\hyp composes into DVBP.

\pagebreak[3]
\begin{problem}[2D vector bin packing (2D-VBP)]
  \label[problem]{dr:prob:2d-vbp}
  \begin{probinstance}
    Requests $\ai{1},\dots,\ai{n}\in(0,1]^2$.
  \end{probinstance}
  \begin{probfind}
    A partition of~$\{1,\dots,n\}$
    into bins~$B_1,\dots,B_k$
    with minimum~$k$ such that,
    for each bin~$B_j$, it holds that (component\hyp wise)
    \[\sum_{i\in B_j}\ai{i}\leq 1.\]
  \end{probfind}
\end{problem}

\noindent
\citet[Theorem~6]{Ray21} proved that
$600/599$\hyp gap 2D-VBP
does not have polynomial\hyp time algorithms with
asymptotic
approximation ratio %
better than $600/599$.
In fact,
he shows that $600/599$\hyp gap 2D-VBP is NP\hyp hard.
We show that
NP\hyp hardness also holds for the following problem variant,
in which all numbers are integer and bounded by a polynomial
in the number of items.

\begin{problem}[Poly\hyp weight 2D-VBP]
  \label[problem]{dr:prob:2d-vbp-special}
  \begin{probinstance}
    Bin capacity~$b\in \N^2$ and
    requests $\ai{1},\dots,\ai{n}\in\N^2$,
    all bounded by a polynomial in~$n$.
  \end{probinstance}
  \begin{probfind}
    A partition of~$\{1,\dots,n\}$
    into bins~$B_1,\dots,B_k$
    with minimum~$k$ such that,
    for each bin $B_j$, it holds that (component\hyp wise)
    \[\sum_{i\in B_j}\ai{i}\leq b.\]
  \end{probfind}
\end{problem}


\begin{lemma}
    \label{dr:lemma:2d-vbp-is-hard}
    The $600/599$\hyp gap
    version of \cref{dr:prob:2d-vbp-special}
    is NP-hard.
\end{lemma}

\begin{proof}
  The NP\hyp hard
  maximum 3-dimensional matching problem (MAX-3-DM) is,
  given three sets
  $X=\{x_1,x_2\dots,x_q\}$, $Y=\{y_1,y_2\dots ,y_q\}$,
  and $Z=\{z_1,z_2,\dots,\allowbreak z_q\}$,
  and a set of triples~$T\subseteq X\times Y\times Z$,
  to find a set $T'\subseteq T$ of maximum cardinality
  such that each element of $X$, $Y$, and~$Z$
  occurs in at most one triple in~$T'$.
  \citet{Ray21} showed a reduction
  of MAX-3-DM to $600/599$\hyp gap 2D-VBP. 
  We modify his reduction so that it outputs instances of
  \cref{dr:prob:2d-vbp-special} instead,
  without changing optimum number of bins.
  
  To this end,
  we first briefly describe the reduction,
  omitting the correctness proof.
  Let $r:=64q$,
  and, for
  each $x_i\in X$ let $x'_i:=ir+1$,
  for each $y_i\in Y$ let $y'_i:=ir^2+2$,
  for each $z_i\in Z$ let $z'_i:=ir^3+4$,
  and for each triple $(x_i,y_j,z_k)\in T$,
  let $t'_{(i,j,k)}=r^4-kr^3-jr^2-ir+8$.
  Call the set of all these integers~$U'$.
  The key feature of this construction is that
  the four numbers~$x'_i$, $y'_j$, $z'_k$, and~$t'_{(i, j, k)}$
  sum up to exactly $b:=r^4+15$.
  The final 2D-VBP instance consists of the requests
  \[\Biggl(\frac15+\frac{a'}{5b},\frac3{10}-\frac{a'}{5b}\Biggr)\text{\quad for each }a'\in U'\]
  and at most $|T|+3q$ (the exact number is a technical detail)
  additional
  \emph{dummy} requests of the form
  \[\Biggl(\frac35,\frac35\Biggr).\]
  The bin capacity, by definition of 2D-VBP,
  can be considered as $(1,1)$.
  The key observation here is that a bin can fit four requests
  if and only if the requests
  were built from the integers~$x'_i$, $y'_j$, $z'_k$, and~$t'_{(i, j, k)}$,
  or, in other words, if the requests correspond to some triple in~$T$.

  All requests
  are vectors of rational numbers
  whose denominators have a common multiple
  \[D:=10\cdot (r^4+15)=10\cdot ((64q)^4+15)\]
  and whose numerators are all bounded by a polynomial of~$q$.
  Multiplying all requests by~$D$
  and setting a bin capacity of~$(D,D)$,
  we get an equivalent instance of \cref{dr:prob:2d-vbp-special}
  in which all numbers are natural
  and bounded by a polynomial in~$q$.
  Since the number~$n$ of requests in it is at least~$3q$,
  all numbers are bounded polynomially in the number~$n$ of requests.
  Thus,
  Ray's result holds for \cref{dr:prob:2d-vbp-special}.
\end{proof}

\noindent
\cref{dr:lemma:2d-vbp-is-hard}
makes \cref{dr:prob:2d-vbp-special} fundamentally different
from standard Bin Packing,
whose 3/2-gap version is NP\hyp hard,
but only weakly so ---
if the item sizes are bounded by a polynomial
in the number of items,
standard Bin Packing is polynomial\hyp time solvable.

The polynomial bound on the request sizes plays a crucial role
in bounding the number of request flavors~$\numflavors$ in the following result:

\begin{lemma}
    \label{dr:lemma:alpha-gap-and-composition}
    The $\alpha$\hyp gap version of \cref{dr:prob:2d-vbp-special}
    AND\hyp composes into $\alpha$\hyp gap DVBP
    parameterized by $\mc+\numflavors$.
\end{lemma}

\begin{proof}
    Assume instances~$I_1,I_2,\dots,I_s$ of the
    \alphagap{} version of \cref{dr:prob:2d-vbp-special}.
	Without loss of generality,
	we may assume that each of them consists of the same number~$n$
	of requests,
	each is asking for the same number~$r$ of bins
    and each instance's bin size is~$b$
	(since this is a polynomial equivalence relation).
	As a result,
	the entries of all vectors are bounded by $\poly(n)$.

	One now can create a DVBP instance~$I$
        consisting of all requests
        of all instances~$I_i$ for $i\in\{1,\dots,s\}$,
        assigning the requests of instance~$I_i$
	the start time~$2i-1$ and the end time~$2i$.
	Obviously,
	the created DVBP instance~$I$
	can %
	fit into $r$~bins if and only if each
	of the instances~$I_i$ can fit into $r$~bins.
	Moreover,
	if
	at least one of the instances~$I_i$ requires more than~$r$ bins,
	then it requires at least $\alpha r$ bins,
	meaning that $I$~requires at least $\alpha r$ bins.

	Also note that the maximum number of requests active
	at any time in~$I$ is $n$ and that
        the number of flavors in~$I$
        is bounded from above by $\poly(n)$.
	That is,
	$\mc+\numflavors \leq \poly(n)\leq\poly(\max_{i=1}^s|I_s|)+\log s$;
	we thus built a valid AND-composition.
\end{proof}

\noindent
Finally,
\cref{dr:thm:no-alpha-gap-polykern} now follows
from \cref{dr:prop:crosscomp,dr:lemma:alpha-gap-and-composition}
and the NP\hyp hardness of
the $\alpha$-gap variant of \cref{dr:prob:2d-vbp-special}
for all $\alpha\leq 600/599$ (\cref{dr:lemma:2d-vbp-is-hard}).

\section{Exact data reduction}
\label{sec:ilp}
\noindent
Due to the results in the previous section,
we do not expect $(1+\varepsilon)$\hyp approximate
preprocessing for DVBP
for arbitrarily small~$\varepsilon>0$
with provable \emph{a priori size} bounds of the output.
This, however,
does not prevent us from designing
data reduction algorithms 
that show \emph{good a posterior} results in experiments.

In this section,
we describe a simple data reduction rule
that maintains optimality of solutions.

\subsection{DVBP with multiplicities}
\noindent
Motivated by the following theorem,
which we prove in this section,
our first data reduction approach
will focus on lowering the number~$\ntypes$ of request types.

\begin{theorem}\label[theorem]{thm:ilp}
  DVBP with $k$~bins and $\ntypes$~request types
  is solvable in $2^{O(\ntypes k\log \ntypes k)}\cdot \poly(n)$~time.
\end{theorem}
\noindent
To prove the theorem,
we follow an approach of \citet{FGR11}
for deriving algorithms for packing problems
parameterized by the ``number of numbers''.
Specifically,
we modify an ILP model for DVBP of \citet[Section~3.1]{DFI20},
replacing Boolean variables by integer variables
so as to allow for requests with multiplicities.
The number of variables in the ILP will be $O(\ntypes k)$,
so that \cref{thm:ilp}
immediately follows from a result of \citet{Len83} and \citet{Kan87}.

To state the ILP,
let $T:=\max_{i=1}^n\di{n}$~be the latest time that any request ends,
$S_t:=\{ j\mid \ri j\leq t<\di j\}$ be the requests active at time~$t\in\{0,\dots,T\}$, and
let $n^{(i)}$~be the number of requests of type $i$, $1\leq i\leq\ntypes$.
The model involves the following variables:
\begin{description}
\item[$x_{ij}$] for $1\leq i\leq \ntypes$ and $1\leq j\leq k$ is
  the number of requests of type~$i$ packed into bin~$j$.
\item[$y_j$] for $1\leq j\leq k$ equals one if bin~$j$ is used
  during at least one time instant and zero otherwise.
\end{description}
Then,
DVBP is modeled using the following ILP.
  \allowdisplaybreaks
\begin{align}
  \min&\sum_{j=1}^ky_j&&\text{s.\,t.}\notag{}\\
  \sum_{j=1}^k x_{ij} &= n^{(i)} &&  1\leq i\leq \ntypes,\label{eq:assign}\\
  x_{ij} & \leq n^{(i)}y_j, && 1\leq i\leq \ntypes, 1\leq j\leq k,\label{eq:activate}\\
  \sum_{i\in S_t}x_{ij}\ai{i}&\leq y_jb, && 1\leq t\leq T, 1\leq j\leq \ntypes \label{eq:resource}\\
  y_j&\in \{0, 1\},&&1\leq j\leq k\notag{}\\
  x_{ij}&\in\mathbb N, &&1\leq i\leq \ntypes, 1\leq j\leq j.\notag{}
\end{align}
Herein,
constraint \eqref{eq:assign} ensures that each request is assigned
to some bin,
constraint \eqref{eq:activate}
makes sure that bin~$j$ is counted as ``used''
whenever some request~$i$ is assigned to it,
and constraint \eqref{eq:resource}
ensures that the bin capacity~$b$ is not exceeded
at any time instant~$t$.

\subsection{Time compression}
\label{dr:timecomp}
\noindent
The number of variables in the above ILP is $O(\ntypes k)$,
whereas the number of constraints is $O(Tr+\ntypes k)$.
It is therefore desirable
to reduce~$\ntypes$ and~$T$.
Since requests of the same flavor and same start and end times
are of the same type,
both reduction of~$\ntypes$ and~$T$ is achieved by the following
simple approach.
  
\begin{proposition}[\citep{BMNW15}]
  \label{dr:thm:tcri}
  In $O(n\log n)$~time,
  the start and end points of all requests
  can be moved to an interval~$\{1,\dots,T\}$
  with minimum~$T$
  maintaining all pairwise intersections
  of request intervals.
\end{proposition}

\noindent
In \cref{dr:exp},
we will see that this data reduction
alone reduces the number of request types by about one third
in real problem instances.

\section{\boldmath$(1+\varepsilon)$\hyp approximate data reduction}
\label{sec:dr}
\label{dr:numa_approx}
\noindent
In this section,
we show a $(1+\varepsilon)$-approximate data reduction algorithm
for DVBP,
that is,
the cost of solutions may increase at most $(1+ \varepsilon)$~times.

In view of the hardness results in \cref{dr:sec:hard},
we do not expect to prove meaningful a priori size bounds of the output.
Indeed, as we will see,
maximizing the data reduction effect of the algorithm is
an NP\hyp hard subproblem itself.
Therefore,
we first describe the high\hyp level algorithm and then,
in subsequent subsections,
describe heuristics for its effective execution.
The data reduction effect will be empirically evaluated
in \cref{dr:exp}.

\subsection{A \texorpdfstring{$(1+\varepsilon)$}{}-approximate data reduction rule}
\label{dr:numa_approx1}
\noindent
Let $L$~be
a lower bound on the number of bins required
to accommodate all requests.
For example,
assuming that all request start and end times
are in $\{1,\dots,2n\}$
by \cref{dr:thm:tcri}, 
one can take
\begin{align}
S_t&:=\{ j\mid \ri j\leq t<\di j\},\notag{}\\
L&:=\max_{t\in\{1,\dots,2n\}}\max_{i\in\{1,\dots,d\}} \sum_{j\in S_t}^n \frac{\ai j_i}{b_i}.\label{dr:L}
\end{align}

\noindent

\begin{rrule}
  {\label[rrule]{dr:rrule:epsL}}
  Delete an arbitrary set of requests
  that can be packed into $\lfloor\varepsilon L\rfloor$ bins.
\end{rrule}

\begin{theorem}
  \cref{dr:rrule:epsL} is a
  $(1+\varepsilon)$\hyp preprocessing for DVBP.
\end{theorem}
\begin{proof}
  Let $I$~be the DVBP instance before
  and $I'$~be the DVBP instance after applying Reduction Rule~\ref{dr:rrule:epsL}.
  If all requests of~$I$ can be packed into $k$~bins,
  then the requests of~$I'$ also can.
  Thus
  \(
  \OPT(I')\leq\OPT(I).
  \)
  If all requests of~$I'$ can be packed into $k$~bins,
  then the requests of~$I$ can be packed into $k+\lfloor\varepsilon L\rfloor$~bins.
  Thus,
  any $\alpha$\hyp approximate solution for~$I'$
  can be turned into an $\alpha(1+\varepsilon)$\hyp approximate
  solution for~$I$ since
  \begin{align*}
    \alpha\OPT(I')+\varepsilon L&\leq \alpha\OPT(I)+\varepsilon 
                                  \OPT(I)\\
    &\leq \alpha(1+\varepsilon)\OPT(I).\qedhere
  \end{align*}
\end{proof}

\noindent
In order to delete
the maximum number of requests using \cref{dr:rrule:epsL},
one has to pack a maximum number of requests into $\lfloor\varepsilon L\rfloor$~bins,
which is an NP\hyp hard task.
One possible implementation of \cref{dr:rrule:epsL}
is solving this task heuristically.
Indeed,
our experiments in \cref{dr:exp} show that,
solving the task exactly would not significantly
increase the effect of \cref{dr:rrule:epsL}
on our real\hyp world data.

\subsection{Greedily packing $\varepsilon L$ bins}
\label{dr:one_numa_bin}
\noindent
In order to pack $\varepsilon L$~bins,
we repeatedly apply a greedy algorithm for packing one bin.

In order to pack one bin,
we assign each request a priority,
then iterate over all not yet packed
requests by non\hyp increasing priority
and,
if the request still fits into the bin,
we pack it.
In detail, the procedure is described in \cref{algo:greed},
where the array~$S$ maintains the available bin space at each time instant
and $A$~is the set of requests packed into the bin.

\begin{algorithm}[t]
    \SetKwInOut{KwIn}{Input}
    \SetKwInOut{KwOut}{Result}
    \KwIn{A bin capacity $b\in\N^d$,
          \emph{requests}~$\ai {1},\dots,\ai {n}\in \N^d$
          with \break start times~$\ri 1,\dots,\ri n\in\{1,\dots,T\}$,
          \break end times~$\di 1,\dots,\di n\in\{1,\dots,T\}$,
          \break and priorities~$f:\{1,\dots,n\}\mapsto\mathbb{Q}$.}
    \KwOut{Indices of packed requests.}

    $A\gets \emptyset$\;
    $S\gets [b,b,\dots,b]$ of size~$T$\;

    \For{$\ai{i}$ in non-increasing order of  $f(i)$} {\label{algo:greed:loop}
        \If{$\ai{i}\leq S[t]$ for $t\in\{\ri{i},\dots,\di{i}-1\}$} {
            $A\gets A\cup\{i\}$\;
            \For{$t\in\{\ri{i},\dots,\di{i}-1\}$} {
                $S[t]\gets S[t]-\ai{i}$\;
            }
        }
    }
    \Return{$A$}\;

    \caption{Greedy packing into one bin.}\label{algo:greed}
\end{algorithm}

The priorities 
are chosen as follows.
Since the goal of the algorithm
is to pack as many requests as possible,
it makes sense to first pack those requests
that require less %
space and block less time instants.
A convenient value to use
as an indicator of a request's small size
is the number of requests of the same flavor
as $\ai{i}$ that can fit into one bin.
Thus,
a canonical priority assignment to each request~$i$ is
\[\alpha(i):=\min_{j\in\{1,\dots,d\}}\Biggl\lfloor \frac{b_j}{\ai{i}_j}\Biggr\rfloor.\]
From $\alpha(i)$,
we also derive the following priorities:
\[f_1(i):=\alpha(i)+\frac{1}{2(\di{i}-\ri{i})},\]
which uses the time span as a tie-breaker for $\alpha(i)$
and places the shorter spanned requests earlier;
and
\[f_2(i):=\alpha(i)\cdot \frac{T}{\di{i}-\ri{i}},\]
which is an upper bound %
on the number of requests of the same flavor
and life time as $\ai{i}$ that can be placed
over all time instants.

We found that~$f_2$ works
the best among $\alpha,f_1,f_2$,
which we evaluated as follows.

\subsection{Quality estimation of greedy packing}
\label{dr:quality-estimation}
\noindent
In the following,
we describe how to measure the effectivity
of our greedy packing algorithm,
thus ultimately answering the question
of how effectively we use \cref{dr:rrule:epsL}
and whether the effect of \cref{dr:rrule:epsL}
can be significantly increased
by solving the problem of packing into $\varepsilon L$
in a more sophisticated way,
or even optimally.

In order to estimate the effect of \cref{dr:rrule:epsL}
in comparison to its potential effect, we use
two \emph{utilization} values
comparing the number~$n$ of requests
in the input instance
to the number~$n'$ of requests in the output instance
(that is, $n-n'$ requests are deleted by \cref{dr:rrule:epsL}).
The first value is
\begin{align}
R&:=\frac{n-n'}{U}\leq1,\label{eq:R}
\end{align}
where~$U$ is an upper bound on
the number of requests removable by \cref{dr:rrule:epsL}.
The closer~$R$ is to~$1$,
the closer the data reduction effect is to the maximum possible.
The second value is
\begin{align}
  K&:=\frac{n'}{n-U}\geq1,
     \label{eq:K}
\end{align}
where $n-U$ is a lower bound on
the number of remaining requests~$n'$.
Low values of~$K$ mean
that the size of the reduced instance
is close to what it could be if we packed
the $\lfloor\varepsilon L\rfloor$ bins optimally.

\begin{algorithm}[t]
    \SetKwInOut{KwIn}{Input}
    \SetKwInOut{KwOut}{Result}
    \KwIn{A bin capacity $b\in\N^d$,
      \break
          \emph{requests} $\ai {1},\dots,\ai {n}\in \N^d$
          with
          \break start times $\ri 1,\dots,\ri n\in\{1,\dots,T\}$,
          \break end times $\di 1,\dots,\di n\in\{1,\dots,T\}$,
          \break and the number of bins $k\in\N$.}
    \KwOut{An upper bound $U$.}
    \SetKwFor{RepeatNTimes}{repeat}{times}{end}

    $S\gets{}$size-$d$-array of empty lists\;

    \For{any request type $(\ai{*},\ri{*},\di{*})$ of any multiplicity $m$} {\label{algo:removed-ub:loop-types}
        $\displaystyle p\gets\min\,\{m\}\cup\Biggl\{k\cdot\Biggl\lfloor \frac{b_j}{\ai{*}_j}\Biggr\rfloor\Biggm|j\in\{1,\dots,d\}\Biggr\}$\;\label{algo:removed-ub:fitting}
        \For{$j\in\{1,\dots,d\}$} {
            append $p$ times $\ai{*}_j$ to $S[j]$\;\label{algo:removed-ub:put}
        }
    }
    \For{$j\in\{1,\dots,d\}$} {
        sort $S[j]$ in non-decreasing order\;\label{algo:removed-ub:sort}
        $\displaystyle U_j\gets\max\,\Biggl\{q\leq n\Biggm|\sum_{t=1}^qS[j][t]\leq b_j\Biggr\}$\;\label{algo:removed-ub:find-q}
        \label{algo:removed-ub:relax}
    }

    \Return{$\min\,\{U_j\mid j\in\{1,\dots,d\}\}$}\;
    \caption{Computation of upper bound~$U$}\label{algo:removed-ub}
\end{algorithm}

After reducing an instance,
we can evaluate the effect of \cref{dr:rrule:epsL}
quantified by~$R$ and~$K$,
using the upper bound~$U$.
One way to obtain~$U$
is to apply %
\cref{algo:removed-ub} with $k=\lfloor\varepsilon L\rfloor$;
it works as follows.

In \cref{algo:removed-ub:loop-types},
it iterates over all request types,
and, for each request type and its multiplicity~$m$,
in \cref{algo:removed-ub:fitting},
computes the maximum number~$p$
of requests of this type that fit into $k$ bins.
In \cref{algo:removed-ub:put},
for each $j\in\{1,\dots,d\}$,
it builds a list~$S[j]$
containing $j$-th components of
this request time $p$~times,
because the other $m-p$ requests of the same type will
not fit into the bins anyway.
In \cref{algo:removed-ub:find-q},
for each $j\in\{1,\dots,d\}$,
it computes an upper bound~$U_j$
by greedily packing
the requests
with the smallest value in the $j$-th component first,
ignoring the   other components.
Finally,
in \cref{algo:removed-ub:relax},
the smallest of the upper bounds is returned.

\begin{table*}[t]
          \caption{Data reduction effect of \cref{dr:thm:tcri} and \cref{dr:rrule:epsL} with $\varepsilon=0.05$.
          Herein, $L$~is the lower bound \eqref{dr:L},
          $T$~is the number of time instants,
          $n$~is the number of requests,
          $\ntypes$ is the number of types,
          $R$ and~$K$ are the utilization measures \eqref{eq:R} and \eqref{eq:K},
          respectively.
      }
	\centering
	\begin{tabular}{ll|rrrrrr}
          \toprule
          &data set                             & $L$  & $T$        & $n$  & $\ntypes$  & $R$ & $K$ \\
		\midrule
        \multirow{3}{*}{\rotatebox{90}{\huawei{}}}  & initial                           & 814  & 152\,366 & 125\,430 & 111\,050 & & \\
         & after \cref{dr:thm:tcri}   &      &  29\,347 & 125\,430 &  78\,010 & & \\
         & after \cref{dr:rrule:epsL} &      &      356 &  10\,079 &   1\,798 & 0.969 & 1.576 \\
		\midrule
          \multirow{3}{*}{\rotatebox{90}{\azure{}}}    & initial                         & 116864 & 4\,650\,911 & 5\,559\,800 & 3\,792\,136 & & \\
           & after \cref{dr:thm:tcri}    &        &    695\,408 & 5\,559\,800 & 2\,271\,582 & & \\
         & after \cref{dr:rrule:epsL}  &        &     20\,226 & 1\,069\,118 &     81\,357 & 0.946 & 1.319 \\
		\bottomrule
	\end{tabular}
\label[table]{table:dataset-stats}
\end{table*}
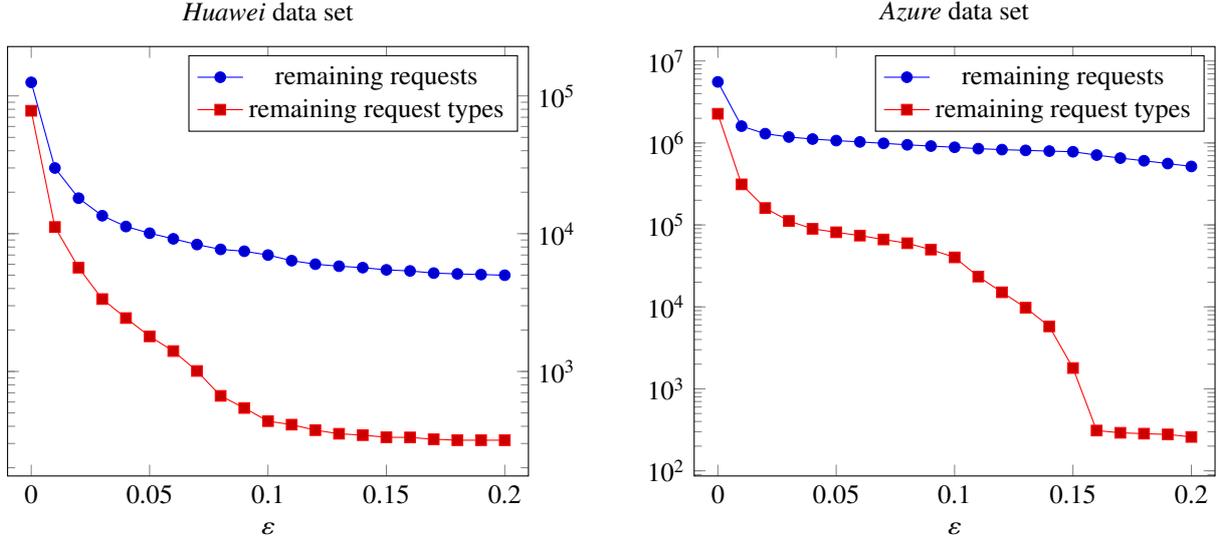
\begin{figure*}[t]
    \begin{tikzpicture}%
      \begin{axis}[xlabel=$\varepsilon$,
        title={\huawei{} data set},
          xtick={0,0.05,0.10,0.15,0.20},
          ylabel near ticks, yticklabel pos=right,
          xticklabel style={/pgf/number format/fixed,/pgf/number format/precision=2},
          ymode=log,
          xmin=-0.01,xmax=0.21
          ]
          \addplot table[col sep=comma,x expr=\thisrow{eps}*1e-2,y=join_num_items]{
initial_max_time,initial_max_clique,initial_num_items,initial_num_entries,initial_num_components,initial_num_entry_classes,shrinked_max_time,shrinked_max_clique,shrinked_num_items,shrinked_num_entries,shrinked_num_components,shrinked_num_entry_classes,eps,lb,num_servers,one_by_one_max_time,one_by_one_max_clique,one_by_one_num_items,one_by_one_num_entries,one_by_one_num_components,one_by_one_num_entry_classes,removed_utilization,left_utilization,join_max_time,join_max_clique,join_num_items,join_num_entries,join_num_components,join_num_entry_classes
152365,9496,125430,125430,1,111050,29346,9496,125430,78010,1,78010,0,814,0,29346,9496,125430,78010,1,78010,-nan,1,29346,9496,125430,78010,1,78010
152365,9496,125430,125430,1,111050,29346,9496,125430,78010,1,78010,1,814,8,2672,9188,29980,11177,1,11177,0.83778,2.60741,2672,9188,29980,11177,1,11177
152365,9496,125430,125430,1,111050,29346,9496,125430,78010,1,78010,2,814,16,1381,9031,18107,5669,1,5669,0.922209,1.99989,1381,9031,18107,5669,1,5669
152365,9496,125430,125430,1,111050,29346,9496,125430,78010,1,78010,3,814,24,700,8911,13500,3354,1,3354,0.952677,1.70025,700,8911,13500,3354,1,3354
152365,9496,125430,125430,1,111050,29346,9496,125430,78010,1,78010,4,814,32,479,8481,11270,2441,1,2441,0.96476,1.58732,479,8481,11270,2441,1,2441
152365,9496,125430,125430,1,111050,29346,9496,125430,78010,1,78010,5,814,40,355,8114,10079,1798,1,1798,0.969059,1.57583,355,8114,10079,1798,1,1798
152365,9496,125430,125430,1,111050,29346,9496,125430,78010,1,78010,6,814,48,265,7868,9174,1406,1,1406,0.972178,1.56901,265,7868,9174,1406,1,1406
152365,9496,125430,125430,1,111050,29346,9496,125430,78010,1,78010,7,814,56,179,7651,8340,1008,1,1008,0.975425,1.54731,179,7651,8340,1008,1,1008
152365,9496,125430,125430,1,111050,29346,9496,125430,78010,1,78010,8,814,65,115,7233,7693,665,1,665,0.977119,1.55855,115,7233,7693,665,1,665
152365,9496,125430,125430,1,111050,29346,9496,125430,78010,1,78010,9,814,73,91,7115,7457,543,1,543,0.976307,1.6232,91,7115,7457,543,1,543
152365,9496,125430,125430,1,111050,29346,9496,125430,78010,1,78010,10,814,81,70,6703,6997,436,1,436,0.977791,1.62456,70,6703,6997,436,1,436
152365,9496,125430,125430,1,111050,29346,9496,125430,78010,1,78010,11,814,89,64,6080,6370,411,1,411,0.980895,1.57245,64,6080,6370,411,1,411
152365,9496,125430,125430,1,111050,29346,9496,125430,78010,1,78010,12,814,97,56,5726,6005,375,1,375,0.982049,1.57117,56,5726,6005,375,1,375
152365,9496,125430,125430,1,111050,29346,9496,125430,78010,1,78010,13,814,105,53,5537,5806,353,1,353,0.981925,1.61099,53,5537,5806,353,1,353
152365,9496,125430,125430,1,111050,29346,9496,125430,78010,1,78010,14,814,113,52,5407,5669,345,1,345,0.98135,1.67079,52,5407,5669,345,1,345
152365,9496,125430,125430,1,111050,29346,9496,125430,78010,1,78010,15,814,122,50,5200,5461,333,1,333,0.981173,1.72871,50,5200,5461,333,1,333
152365,9496,125430,125430,1,111050,29346,9496,125430,78010,1,78010,16,814,130,50,5104,5365,332,1,332,0.980323,1.81557,50,5104,5365,332,1,332
152365,9496,125430,125430,1,111050,29346,9496,125430,78010,1,78010,17,814,138,48,4915,5174,322,1,322,0.980337,1.87328,48,4915,5174,322,1,322
152365,9496,125430,125430,1,111050,29346,9496,125430,78010,1,78010,18,814,146,48,4856,5107,317,1,317,0.979422,1.98022,48,4856,5107,317,1,317
152365,9496,125430,125430,1,111050,29346,9496,125430,78010,1,78010,19,814,154,48,4800,5051,317,1,317,0.978619,2.08633,48,4800,5051,317,1,317
152365,9496,125430,125430,1,111050,29346,9496,125430,78010,1,78010,20,814,162,48,4744,4995,317,1,317,0.977978,2.18791,48,4744,4995,317,1,317
};
\addlegendentry{remaining requests}
\addplot table[col sep=comma,x expr=\thisrow{eps}*1e-2,y=join_num_entries]{
initial_max_time,initial_max_clique,initial_num_items,initial_num_entries,initial_num_components,initial_num_entry_classes,shrinked_max_time,shrinked_max_clique,shrinked_num_items,shrinked_num_entries,shrinked_num_components,shrinked_num_entry_classes,eps,lb,num_servers,one_by_one_max_time,one_by_one_max_clique,one_by_one_num_items,one_by_one_num_entries,one_by_one_num_components,one_by_one_num_entry_classes,removed_utilization,left_utilization,join_max_time,join_max_clique,join_num_items,join_num_entries,join_num_components,join_num_entry_classes
152365,9496,125430,125430,1,111050,29346,9496,125430,78010,1,78010,0,814,0,29346,9496,125430,78010,1,78010,-nan,1,29346,9496,125430,78010,1,78010
152365,9496,125430,125430,1,111050,29346,9496,125430,78010,1,78010,1,814,8,2672,9188,29980,11177,1,11177,0.83778,2.60741,2672,9188,29980,11177,1,11177
152365,9496,125430,125430,1,111050,29346,9496,125430,78010,1,78010,2,814,16,1381,9031,18107,5669,1,5669,0.922209,1.99989,1381,9031,18107,5669,1,5669
152365,9496,125430,125430,1,111050,29346,9496,125430,78010,1,78010,3,814,24,700,8911,13500,3354,1,3354,0.952677,1.70025,700,8911,13500,3354,1,3354
152365,9496,125430,125430,1,111050,29346,9496,125430,78010,1,78010,4,814,32,479,8481,11270,2441,1,2441,0.96476,1.58732,479,8481,11270,2441,1,2441
152365,9496,125430,125430,1,111050,29346,9496,125430,78010,1,78010,5,814,40,355,8114,10079,1798,1,1798,0.969059,1.57583,355,8114,10079,1798,1,1798
152365,9496,125430,125430,1,111050,29346,9496,125430,78010,1,78010,6,814,48,265,7868,9174,1406,1,1406,0.972178,1.56901,265,7868,9174,1406,1,1406
152365,9496,125430,125430,1,111050,29346,9496,125430,78010,1,78010,7,814,56,179,7651,8340,1008,1,1008,0.975425,1.54731,179,7651,8340,1008,1,1008
152365,9496,125430,125430,1,111050,29346,9496,125430,78010,1,78010,8,814,65,115,7233,7693,665,1,665,0.977119,1.55855,115,7233,7693,665,1,665
152365,9496,125430,125430,1,111050,29346,9496,125430,78010,1,78010,9,814,73,91,7115,7457,543,1,543,0.976307,1.6232,91,7115,7457,543,1,543
152365,9496,125430,125430,1,111050,29346,9496,125430,78010,1,78010,10,814,81,70,6703,6997,436,1,436,0.977791,1.62456,70,6703,6997,436,1,436
152365,9496,125430,125430,1,111050,29346,9496,125430,78010,1,78010,11,814,89,64,6080,6370,411,1,411,0.980895,1.57245,64,6080,6370,411,1,411
152365,9496,125430,125430,1,111050,29346,9496,125430,78010,1,78010,12,814,97,56,5726,6005,375,1,375,0.982049,1.57117,56,5726,6005,375,1,375
152365,9496,125430,125430,1,111050,29346,9496,125430,78010,1,78010,13,814,105,53,5537,5806,353,1,353,0.981925,1.61099,53,5537,5806,353,1,353
152365,9496,125430,125430,1,111050,29346,9496,125430,78010,1,78010,14,814,113,52,5407,5669,345,1,345,0.98135,1.67079,52,5407,5669,345,1,345
152365,9496,125430,125430,1,111050,29346,9496,125430,78010,1,78010,15,814,122,50,5200,5461,333,1,333,0.981173,1.72871,50,5200,5461,333,1,333
152365,9496,125430,125430,1,111050,29346,9496,125430,78010,1,78010,16,814,130,50,5104,5365,332,1,332,0.980323,1.81557,50,5104,5365,332,1,332
152365,9496,125430,125430,1,111050,29346,9496,125430,78010,1,78010,17,814,138,48,4915,5174,322,1,322,0.980337,1.87328,48,4915,5174,322,1,322
152365,9496,125430,125430,1,111050,29346,9496,125430,78010,1,78010,18,814,146,48,4856,5107,317,1,317,0.979422,1.98022,48,4856,5107,317,1,317
152365,9496,125430,125430,1,111050,29346,9496,125430,78010,1,78010,19,814,154,48,4800,5051,317,1,317,0.978619,2.08633,48,4800,5051,317,1,317
152365,9496,125430,125430,1,111050,29346,9496,125430,78010,1,78010,20,814,162,48,4744,4995,317,1,317,0.977978,2.18791,48,4744,4995,317,1,317
};
          \addlegendentry{remaining request types}
        \end{axis}
    \end{tikzpicture}\hfill{}
    \begin{tikzpicture}%
        \begin{axis}[xlabel=$\varepsilon$, %
          xtick={0,0.05,0.10,0.15,0.20},
          title={\azure{} data set},
                     xticklabel style={/pgf/number format/fixed,/pgf/number format/precision=2},
                     ymax=15000000,
                     xmin=-0.01,xmax=0.21,
                     ymode=log,
                    ]
                    \addplot table[col sep=comma,x expr=\thisrow{eps}*1e-2,y=join_num_items]{
initial_max_time,initial_max_clique,initial_num_items,initial_num_entries,initial_num_components,initial_num_entry_classes,shrinked_max_time,shrinked_max_clique,shrinked_num_items,shrinked_num_entries,shrinked_num_components,shrinked_num_entry_classes,eps,lb,num_servers,one_by_one_max_time,one_by_one_max_clique,one_by_one_num_items,one_by_one_num_entries,one_by_one_num_components,one_by_one_num_entry_classes,removed_utilization,left_utilization,join_max_time,join_max_clique,join_num_items,join_num_entries,join_num_components,join_num_entry_classes
4650910,885538,5559800,5559800,1,3792136,695407,885538,5559800,2271669,1,2271582,0,116864,0,695407,885538,5559800,2271669,1,2271582,-nan,1,695407,885538,5559800,2271582,1,2271582
4650910,885538,5559800,5559800,1,3792136,695407,885538,5559800,2271669,1,2271582,1,116864,1168,88197,869093,1603981,312808,1,312683,0.951266,1.14462,88197,869093,1603981,312683,1,312683
4650910,885538,5559800,5559800,1,3792136,695407,885538,5559800,2271669,1,2271582,2,116864,2337,41878,855364,1297809,160814,1,160673,0.958472,1.16589,41878,855364,1297809,160673,1,160673
4650910,885538,5559800,5559800,1,3792136,695407,885538,5559800,2271669,1,2271582,3,116864,3505,28526,846431,1182221,111937,1,111789,0.957872,1.19454,28526,846431,1182221,111789,1,111789
4650910,885538,5559800,5559800,1,3792136,695407,885538,5559800,2271669,1,2271582,4,116864,4674,22198,838602,1116159,89788,1,89650,0.951337,1.25572,22198,838602,1116159,89650,1,89650
4650910,885538,5559800,5559800,1,3792136,695407,885538,5559800,2271669,1,2271582,5,116864,5843,20225,831836,1069118,81489,1,81357,0.945505,1.31942,20225,831836,1069118,81357,1,81357
4650910,885538,5559800,5559800,1,3792136,695407,885538,5559800,2271669,1,2271582,6,116864,7011,18553,827056,1030464,74225,1,74102,0.94152,1.37554,18553,827056,1030464,74102,1,74102
4650910,885538,5559800,5559800,1,3792136,695407,885538,5559800,2271669,1,2271582,7,116864,8180,16828,821255,989904,66569,1,66450,0.939289,1.42529,16828,821255,989904,66450,1,66450
4650910,885538,5559800,5559800,1,3792136,695407,885538,5559800,2271669,1,2271582,8,116864,9349,15227,814853,949332,59899,1,59783,0.938549,1.46623,15227,814853,949332,59783,1,59783
4650910,885538,5559800,5559800,1,3792136,695407,885538,5559800,2271669,1,2271582,9,116864,10517,12593,810599,916441,50115,1,50009,0.936339,1.52551,12593,810599,916441,50009,1,50009
4650910,885538,5559800,5559800,1,3792136,695407,885538,5559800,2271669,1,2271582,10,116864,11686,9943,807470,885970,40301,1,40204,0.933833,1.5969,9943,807470,885970,40204,1,40204
4650910,885538,5559800,5559800,1,3792136,695407,885538,5559800,2271669,1,2271582,11,116864,12855,5622,802925,852352,23542,1,23452,0.933703,1.64514,5622,802925,852352,23452,1,23452
4650910,885538,5559800,5559800,1,3792136,695407,885538,5559800,2271669,1,2271582,12,116864,14023,3530,798459,829503,15167,1,15078,0.932566,1.70171,3530,798459,829503,15078,1,15078
4650910,885538,5559800,5559800,1,3792136,695407,885538,5559800,2271669,1,2271582,13,116864,15192,2140,792918,811866,9864,1,9775,0.930882,1.7675,2140,792918,811866,9775,1,9775
4650910,885538,5559800,5559800,1,3792136,695407,885538,5559800,2271669,1,2271582,14,116864,16360,1220,787127,795908,5849,1,5781,0.929463,1.83229,1220,787127,795908,5781,1,5781
4650910,885538,5559800,5559800,1,3792136,695407,885538,5559800,2271669,1,2271582,15,116864,17529,346,779366,781324,1828,1,1795,0.928075,1.90104,346,779366,781324,1795,1,1795
4650910,885538,5559800,5559800,1,3792136,695407,885538,5559800,2271669,1,2271582,16,116864,18698,25,711800,711926,339,1,310,0.937298,1.83667,25,711800,711926,310,1,310
4650910,885538,5559800,5559800,1,3792136,695407,885538,5559800,2271669,1,2271582,17,116864,19866,21,653801,653922,319,1,291,0.944248,1.79522,21,653801,653922,291,1,291
4650910,885538,5559800,5559800,1,3792136,695407,885538,5559800,2271669,1,2271582,18,116864,21035,19,606649,606766,311,1,284,0.949053,1.78001,19,606649,606766,284,1,284
4650910,885538,5559800,5559800,1,3792136,695407,885538,5559800,2271669,1,2271582,19,116864,22204,19,560546,560663,305,1,278,0.95373,1.76237,19,560546,560663,278,1,278
4650910,885538,5559800,5559800,1,3792136,695407,885538,5559800,2271669,1,2271582,20,116864,23372,16,517028,517133,286,1,259,0.958748,1.72284,16,517028,517133,259,1,259
};
                    \addlegendentry{remaining requests};
                    \addplot table[col sep=comma,x expr=\thisrow{eps}*1e-2,y=join_num_entries]{
initial_max_time,initial_max_clique,initial_num_items,initial_num_entries,initial_num_components,initial_num_entry_classes,shrinked_max_time,shrinked_max_clique,shrinked_num_items,shrinked_num_entries,shrinked_num_components,shrinked_num_entry_classes,eps,lb,num_servers,one_by_one_max_time,one_by_one_max_clique,one_by_one_num_items,one_by_one_num_entries,one_by_one_num_components,one_by_one_num_entry_classes,removed_utilization,left_utilization,join_max_time,join_max_clique,join_num_items,join_num_entries,join_num_components,join_num_entry_classes
4650910,885538,5559800,5559800,1,3792136,695407,885538,5559800,2271669,1,2271582,0,116864,0,695407,885538,5559800,2271669,1,2271582,-nan,1,695407,885538,5559800,2271582,1,2271582
4650910,885538,5559800,5559800,1,3792136,695407,885538,5559800,2271669,1,2271582,1,116864,1168,88197,869093,1603981,312808,1,312683,0.951266,1.14462,88197,869093,1603981,312683,1,312683
4650910,885538,5559800,5559800,1,3792136,695407,885538,5559800,2271669,1,2271582,2,116864,2337,41878,855364,1297809,160814,1,160673,0.958472,1.16589,41878,855364,1297809,160673,1,160673
4650910,885538,5559800,5559800,1,3792136,695407,885538,5559800,2271669,1,2271582,3,116864,3505,28526,846431,1182221,111937,1,111789,0.957872,1.19454,28526,846431,1182221,111789,1,111789
4650910,885538,5559800,5559800,1,3792136,695407,885538,5559800,2271669,1,2271582,4,116864,4674,22198,838602,1116159,89788,1,89650,0.951337,1.25572,22198,838602,1116159,89650,1,89650
4650910,885538,5559800,5559800,1,3792136,695407,885538,5559800,2271669,1,2271582,5,116864,5843,20225,831836,1069118,81489,1,81357,0.945505,1.31942,20225,831836,1069118,81357,1,81357
4650910,885538,5559800,5559800,1,3792136,695407,885538,5559800,2271669,1,2271582,6,116864,7011,18553,827056,1030464,74225,1,74102,0.94152,1.37554,18553,827056,1030464,74102,1,74102
4650910,885538,5559800,5559800,1,3792136,695407,885538,5559800,2271669,1,2271582,7,116864,8180,16828,821255,989904,66569,1,66450,0.939289,1.42529,16828,821255,989904,66450,1,66450
4650910,885538,5559800,5559800,1,3792136,695407,885538,5559800,2271669,1,2271582,8,116864,9349,15227,814853,949332,59899,1,59783,0.938549,1.46623,15227,814853,949332,59783,1,59783
4650910,885538,5559800,5559800,1,3792136,695407,885538,5559800,2271669,1,2271582,9,116864,10517,12593,810599,916441,50115,1,50009,0.936339,1.52551,12593,810599,916441,50009,1,50009
4650910,885538,5559800,5559800,1,3792136,695407,885538,5559800,2271669,1,2271582,10,116864,11686,9943,807470,885970,40301,1,40204,0.933833,1.5969,9943,807470,885970,40204,1,40204
4650910,885538,5559800,5559800,1,3792136,695407,885538,5559800,2271669,1,2271582,11,116864,12855,5622,802925,852352,23542,1,23452,0.933703,1.64514,5622,802925,852352,23452,1,23452
4650910,885538,5559800,5559800,1,3792136,695407,885538,5559800,2271669,1,2271582,12,116864,14023,3530,798459,829503,15167,1,15078,0.932566,1.70171,3530,798459,829503,15078,1,15078
4650910,885538,5559800,5559800,1,3792136,695407,885538,5559800,2271669,1,2271582,13,116864,15192,2140,792918,811866,9864,1,9775,0.930882,1.7675,2140,792918,811866,9775,1,9775
4650910,885538,5559800,5559800,1,3792136,695407,885538,5559800,2271669,1,2271582,14,116864,16360,1220,787127,795908,5849,1,5781,0.929463,1.83229,1220,787127,795908,5781,1,5781
4650910,885538,5559800,5559800,1,3792136,695407,885538,5559800,2271669,1,2271582,15,116864,17529,346,779366,781324,1828,1,1795,0.928075,1.90104,346,779366,781324,1795,1,1795
4650910,885538,5559800,5559800,1,3792136,695407,885538,5559800,2271669,1,2271582,16,116864,18698,25,711800,711926,339,1,310,0.937298,1.83667,25,711800,711926,310,1,310
4650910,885538,5559800,5559800,1,3792136,695407,885538,5559800,2271669,1,2271582,17,116864,19866,21,653801,653922,319,1,291,0.944248,1.79522,21,653801,653922,291,1,291
4650910,885538,5559800,5559800,1,3792136,695407,885538,5559800,2271669,1,2271582,18,116864,21035,19,606649,606766,311,1,284,0.949053,1.78001,19,606649,606766,284,1,284
4650910,885538,5559800,5559800,1,3792136,695407,885538,5559800,2271669,1,2271582,19,116864,22204,19,560546,560663,305,1,278,0.95373,1.76237,19,560546,560663,278,1,278
4650910,885538,5559800,5559800,1,3792136,695407,885538,5559800,2271669,1,2271582,20,116864,23372,16,517028,517133,286,1,259,0.958748,1.72284,16,517028,517133,259,1,259
};
                    \addlegendentry{remaining request types};
        \end{axis}
    \end{tikzpicture}
    \caption{Data reduction effect of \cref{dr:thm:tcri} and \cref{dr:rrule:epsL}       for $\varepsilon\in[0,0.2]$ with step size~$0.01$.}
    \label{fig:plots-huawei-cloud}
    \label{fig:plots-azure}
\end{figure*}

\section{Experiments}
\label{dr:exp}
\noindent
We present an experimental evaluation of the effect of our data reduction
on real\hyp world instances for a virtual machine scheduling problem
(\cref{ex:vm-scheduling}).

\paragraph{Data sets}%
\looseness=-1
We evaluated our data reduction approach
on six data sets with similar results.
However,
since four of the data sets are confidential,
in detail
we can present results here only for
two openly available data sets.

The first one is ``Huawei-East-1''
released by Huawei \citep{SCC+22}.\footnote{https://github.com/huaweicloud/VM-placement-dataset}
We denote this data set \huawei{}.
In this data set, %
$d=2$,
that is,
each request has two resource requirements:
the number of requested CPU cores and the amount of allocated RAM.
As suggested by its authors,
we simulated CPU sharing for small virtual machines
(with flavors 2U4G, 4U8G, 8U16G, 1U2G, 4U16G, 1U1G, 2U8G, 8U32G and 1U4G,
where machine $X$U$Y$G has $X$~CPU cores and $Y$~GB of RAM)
by dividing their CPU demand by~3.
We set the bin size to 40 CPU cores and 90 GB of RAM,
following the data set authors' experimental setup.
All arithmetic calculations for this data set
are carried out with 64-bit integers.

The second data set, called \azure{},
is released by Microsoft and
is called ``Azure Traces for Packing 2020'' \citep{HMM+20}.\footnote{https://github.com/Azure/AzurePublicDataset}
This data set has $d=5$,
where each request has requirements on
the number of CPU cores,
amount of RAM, HDD and SSD space,
and network bandwidth.
In the \azure{} dataset,
each resource requirement
is a fractional number in~$[0,1]$
equal to the fraction of the available resource on the servers, %
so the bin size is set to $(1,1,1,1,1)$.
All computations for this data set are done in
double precision floating point arithmetic.

\paragraph{Reduction effect}

In this section,
we present experimental results of the $(1+\varepsilon)$-approximate
Reduction Rule~\ref{dr:rrule:epsL}.
We perform \cref{dr:thm:tcri}
before~\cref{dr:rrule:epsL}
and also during execution of \cref{dr:rrule:epsL}
after each packed and deleted bin.

Allowing an increase in the number of servers
in an optimal solution by no more than $5\,\%$,
we perform \cref{dr:rrule:epsL}
with $\varepsilon=0.05$.
\cref{table:dataset-stats} shows
the results.
When computing the optimal solution lower bound~$L$
according to (\ref{dr:L}),
\cref{dr:rrule:epsL} tries to pack into
$\lfloor\varepsilon L\rfloor=40$ bins for \huawei{}
and into
$\lfloor\varepsilon L\rfloor=5843$ bins for \azure{}.

The data reduction given by \cref{dr:thm:tcri}
reduces the number of time instants by  about~$6$ times.
It does not change the number of requests,
yet reduces the number of request types
by about one third,
since after time compression,
many requests of the same flavor have coinciding start and end times
and are thus of the same type.

After \cref{dr:rrule:epsL} the number of requests
is decreased significantly:~$5.2$ times for \azure{} and~$12.4$ times for \huawei{}.
And even more significant is the change
in the number of request types,
which, for each data set, fell about~$50$ times.

A natural question is whether
the effect of \cref{dr:rrule:epsL} can be increased
by packing the $\lfloor\varepsilon L\rfloor$ bins for deletion better.
The table presents the resource utilization measures~$R$ and~$K$
described in \cref{dr:quality-estimation}.
The value of~$R$ shows
that our implementation of \cref{dr:rrule:epsL}
removes at least $94\,\%$ of all requests that it can remove in principle.
However,
judging by~$K$,
\cref{dr:rrule:epsL} potentially leaves $57\,\%$
more requests in \huawei{} than there have to be
and $32\,\%$ more requests in \azure{} than have to be.
On the one hand,
this means that
there might still be a little room for improvement,
yet certainly not by an order of magnitude.
On the other and,
the number of requests that have to remain
is only a lower bound.
Thus,
it might even be that the $\lfloor\varepsilon L\rfloor$ bins
are already packed optimally.

Finally,
the choice of~$\varepsilon$
allows for a trade\hyp off of solution quality versus output size of the
data reduction,
illustrated in \cref{fig:plots-huawei-cloud}.
In both cases,
we can observe
data reduction by an order of magnitude for $\varepsilon=0.02$.
The number of request types
shrinks by another order of magnitude for $\varepsilon\in[0.08, 0.12]$.

Good candidates for compromise
are $\varepsilon=0.1$ on the \huawei{} data set
and $\varepsilon=0.16$  on the \azure{} data set,
if that would allow for solving the problem and such an error is permissible.

\section{Conclusion}
\noindent
We have studied the potential of data reduction with performance
guarantees for DVBP.
While we have shown obstacles for proving size bounds
of instances after data reduction,
we proved guarantees on the solution quality.

Despite our data reduction rules being very simple
and coming without size bounds,
they proved to be very effective in experiments with real data
from Microsoft Azure and Huawei Cloud,
shrinking problem instances by orders of magnitude.

Unfortunately,
for two the two open sets,
the number of request types after data reduction
is still out of reach for algorithms
with running time exponential in the number of request types,
for example, for the ILP from \cref{thm:ilp}.

\end{document}